\documentclass[Journal,twocolumn,10pt]{IEEEtran}

\usepackage{graphicx}
\usepackage{amsfonts}
\usepackage{amssymb}
\usepackage{amsmath}

\newcommand{\comment}[1]{}

\newtheorem{theorem}{Theorem}[section]
     \newtheorem{lemma}[theorem]{Lemma}
     
     \newtheorem{corollary}[theorem]{Corollary}

     \newcommand{\qed}{\nobreak \ifvmode \relax \else
           \ifdim\lastskip<1.5em \hskip-\lastskip
           \hskip1.5em plus0em minus0.5em \fi \nobreak
           \vrule height0.75em width0.5em depth0.25em\fi}
\begin{document}


\title{ Equalization for Non-Coherent UWB Systems with
Approximate Semi-Definite Programming }

\author{\authorblockN{Xudong Ma \\}
\authorblockA{FlexRealm Silicon Inc., Virginia, U.S.A.\\
Email: xma@ieee.org} }

\maketitle

\begin{abstract}

In this paper, we propose an approximate semi-definite programming
framework for demodulation and equalization of non-coherent
ultra-wide-band communication systems with
inter-symbol-interference. It is assumed that the communication
systems follow non-linear second-order Volterra models. We formulate
the demodulation and equalization problems as semi-definite
programming  problems. We propose an approximate algorithm for
solving the formulated semi-definite programming problems. Compared
with the existing non-linear equalization approaches, such as in
\cite{ma09}, the proposed semi-definite programming formulation and
approximate solving algorithm have low computational complexity and
storage requirements. We show that the proposed algorithm has
satisfactory error probability performance by simulation results.
The proposed non-linear equalization approach can be adopted for a
wide spectrum of non-coherent ultra-wide-band systems, due to the
fact that most non-coherent ultra-wide-band systems with
inter-symbol-interference follow non-linear second-order Volterra
signal models.

\end{abstract}

\section{Introduction}

\label{section_introduction}

Ultra-Wide-Band (UWB) communication systems have attracted much
attention recently. The UWB communication systems have many
advantages including multi-path diversity, low possibilities of
intercept and high location estimation accuracy. However, UWB
systems also present many challenges compared with narrow-band
communication systems. Especially, the communication channels are
frequency selective with a large number of resolvable multi-paths.
Accurate estimation of channel impulse responses is complex and
difficult.

Existing modulation schemes for UWB can be roughly classified into
two categories, coherent modulation schemes and non-coherent
modulation schemes. The coherent schemes include direct-sequence UWB
and multi-band UWB \cite{win00} \cite{saberinia03} \cite{batra04}.
In these schemes, the demodulation usually depends on accurate
estimation of channel impulse responses. The coherent schemes can
achieve higher transmission rates. However, their complexity and
cost are usually high.

Unlike the coherent modulation schemes, in the non-coherent UWB
modulation schemes, the demodulation usually does not depend on full
knowledge of channel impulse responses. Therefore, the difficulty of
channel estimation is largely avoided. The non-coherent schemes
include various differential encoding schemes, and energy detection
based schemes (see for example \cite{ho02} \cite{choi02}
\cite{mo07}).

One difficulty with the non-coherent modulation schemes is that the
signal models are non-linear, if there exists
Inter-Symbol-Interference (ISI) in the systems \cite{witrisal}
\cite{mo07}. The existing linear equalization approaches generally
do not work well for such non-linear ISI. Because the UWB channels
usually have long delay spreads, the approach that increases the
spaces between symbols to avoid ISI, severally limits the achievable
rates, and therefore is not realistic.

In \cite{ma09}, a new non-linear equalization scheme based on
Semi-Definite Programming (SDP) has been proposed. It is shown that
even though the SDP relaxation approach is sub-optimal, the
performance loss is usually negligible. In \cite{ma09}, an
off-the-shelf general-purpose algorithm is adopted to solve the SDP
programming problems.

However, general-purpose SDP solving algorithms may not be suitable
choices for the UWB demodulation and equalization scenarios. First,
the general-purpose algorithms are usually designed to obtain very
accurate optimization solutions. While, in the UWB demodulation and
equalization scenarios, only approximate solutions are needed to
ensure low demodulation errors, because the SDP optimization
solutions are only intermediate results. By relaxing the requirement
on the accuracy of optimization solutions, the computational
complexity can be greatly reduced. Second, the general-purpose SDP
solving algorithms do not utilize problem structures. In fact, the
computational complexity can be largely reduced by utilizing the
structure of the problems.

In this paper, we propose a new iterative algorithm for solving the
SDP programming problems. The proposed algorithm has low
computational complexity and storage requirements, which make it an
attractive choice for low-complexity high-speed implementations.
First, the algorithm can achieve a close approximate solution of the
optimization problem after only a few iterations. Second, during
each iteration, only one optimization problem with much smaller
problem size needs to be solved. More precisely, the problem size is
equal to the number of bits in one signal block, while, the problem
size of the original matrix optimization is proportional to the
square of the number of bits in one signal block. The correctness
and convergence of the algorithm is proven in this paper. We also
show by simulation results that the demodulation and equalization
algorithm has satisfactory error probability performance.

In this paper, we demonstrate the performance of the proposed
non-linear demodulation and equalization scheme on differential UWB
systems. In fact, the proposed algorithm can also be applied on
other non-coherent UWB systems, because many non-coherent UWB
systems have the same non-linear second-order Volterra signal
models. One thing we wish to stress is that certain channel
parameter estimation is needed in the proposed demodulation
algorithm. However, the estimated model is at the symbol level,
rather than at the Nyquist frequency level. The complexity of this
partial channel estimation is acceptable.

The rest of this paper is organized as follows. In Section
\ref{section_system}, we describe the signal model. The SDP problem
formulation is presented in Section
\ref{section_problem_formulation}. We present the proposed
demodulation and equalization algorithm in Section
\ref{section_approximate_sdp}. Numerical results are presented in
Section \ref{sec_numerical}. Conclusions are presented in Section
\ref{sec_conclusion}.

%

%


Notation: we use the symbol $\pmb{\mathcal S}$ to denote the set of
symmetric matrices. Matrices are denoted by upper bold face letters
and column vectors are denoted by lower bold face letters. We use
$\pmb{A}\succeq 0$ to denote that the matrix $\pmb{A}$ is positive
semi-definite. We use $\pmb{a}\geq 0$ to denote that the elements of
the vector $\pmb{a}$ are non-negative. We use $\pmb{A}_{i,j}$ to
denote the element of the matrix $\pmb{A}$ at the $i$-th row and
$j$-th column. We use $\pmb{a}_{i}$ to denote the $i$-th element of
the vector $\pmb{a}$. We use $\pmb{A}^t$ and $\pmb{a}^t$ to denote
the transpose of the matrix $\pmb{A}$ and the vector $\pmb{a}$
respectively. We use $tr(\pmb{A})$ to denote the trace of the matrix
$\pmb{A}$. We use $\pmb{A}\cdot\pmb{B}$ to denote the inner product
of matrices $\pmb{A}$ and $\pmb{B}$, that is
$\pmb{A}\cdot\pmb{B}=tr(\pmb{A}^t\pmb{B})$. The function
$\mbox{sign}(\cdot)$ is defined as,
\begin{align}
\mbox{sign}(x)=\left\{
\begin{array}{ll}
1, &  \mbox{if }x\geq 0, \\
-1, & \mbox{otherwise}.
\end{array}
\right.
\end{align}

\section{Signal Model}
\label{section_system}

\begin{figure}[h]
 \centering
 \includegraphics[width=3in]{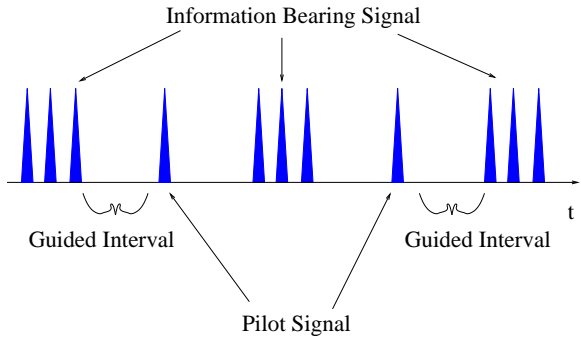}
 \caption{signal is transmitted in a block by block fashion}
 \label{fig_block_transmitting}
\end{figure}

In this paper, we consider the differential UWB systems. We assume
that information is transmitted in a block by block fashion as shown
in Fig. \ref{fig_block_transmitting}. That is, the transmitted
signal $s(t)$ can be written as,
\begin{align}
s(t) = \sum_{k=0}^{\infty} s_k^i(t-kT_b) +\sum_{k=0}^{\infty}
s_k^p(t-kT_b-\tau_{k})
\end{align}
where $s_k^i(t)$ is the signal waveform for the $k$th block of
information bearing signals, and $s_k^p(t)$ is the waveform for the
$k$th block of pilot signals.

The waveform for one block of information bearing signals can be
written as,
\begin{align}
s_k^i(t)=\sum_{n=0}^{N_b-1}\sum_{i=0}^{N_p-1}a_i[n]\bar{w}\left(t-t_i[n]\right)
\end{align}
where $\bar{w}(t)$ is the transmitted pulse, $a_i[n]$ is the pulse
polarity for the $i$-th pulse of the $n$-th symbol, $t_i[n]$ is the
pulse time for the $i$-th pulse of the $n$-th symbol. Each block has
$N_b$ symbols, and each symbol corresponds to $N_p$ pulses.

Denote the data symbol by $d[n]\in\{-1,+1\}$. The data symbols are
differentially encoded as,
\begin{align}
a_i[n]=\left\{
\begin{array}{ll}
a_{N_p-1}[n-1]d[n-1]b_{N_p-1}, & \mbox{ if }i=0 \\
a_{i-1}[n]d[n]b_{i-1}, & \mbox{ otherwise}
\end{array}
\right.
\end{align}
where, $b_0,b_1,\ldots,b_{N_p-1}$ is the pseudo-random amplitude
code sequence, $b_i\in\{-1,+1\}$. The pulse time
\begin{align}
t_i[n]=nT_s+c_i
\end{align}
where $T_s$ is the symbol duration, $c_i$ is the relative pulse
timing.

\comment{The relative pulse timeing $c_i$ is related to the
pseudo-random delay hopping code $\{D_i\}$,
\begin{align}
D_i=\left\{
\begin{array}{ll}
T_s+c_0-c_{N_p-1}, & \mbox{ if }i=N_p-1 \\
c_{i+1}-c_i, & \mbox{ otherwise}
\end{array}
\right.
\end{align}}

The pilot signal $s_k^p(t)$ is introduced to facilitate timing
synchronization and partial channel estimation. Guided intervals are
introduced between blocks of information bearing signals and pilot
signals, so that all inter-block-interference is avoided.

\comment{ Denote the received signal corresponding the $k$th
information bearing signal block by $r_k(t)$. In sequel, we will
drop the notation $k$ if no ambiguity occurs.
 The received signal $r(t)$ can be written as,
\begin{align}
r(t)=\sum_{n=0}^{N_b-1}\sum_{i=0}^{N_p-1}a_i[n]g\left(t-t_i[n]\right)+n(t),
\end{align}
where, $g(t)$ is the channel response for the pulse $\bar{w}(t)$,
$n(t)$ is the noise.}

Similarly as in \cite{ma09}, at the receiver side, an
auto-correlation receiver is used. Denote the received signals
corresponding to one block of information bearing signals by $z[m]$,
$m =1,2,\ldots,N_r$. The signal model of the system is a
second-order Volterra model as follows.
\begin{align}
z[m]=(\pmb{r}+\pmb{P}\pmb{d})^t\pmb{Q}^t\pmb{B}[m]\pmb{Q}(\pmb{r}+\pmb{P}\pmb{d})+\mbox{
noise terms,}
\end{align}
where $\pmb{Q}$, $\pmb{P}$, $\pmb{r}$ are constant matrices and
vectors , and $\pmb{B}[n]$ are matrices that depends on the wireless
channel (more detailed definitions can be found in \cite{ma09}). We
assume that the matrices $\pmb{B}[m]$ can be estimated accurately by
using the pilot signals.

\comment{ The definition of the matrices and vectors are summarized
as follows.
\begin{itemize}
\item \begin{align}
I_g(t_1,t_2;\tau)=\int_{t_1}^{t_2}g(t)g(t+\tau)dt
\end{align}

\item Denote the data vector by $\pmb{d}$,
\begin{align}
\pmb{d}=[d_0,d_1,\ldots,d_{N_b-1}]^t.
\end{align}

\item In the above equation, $\pmb{Q}$ is a diagonal matrix,
\begin{align}
\pmb{Q}=\mbox{diag}[1,b_0,b_0b_1,\ldots],
\end{align}
\begin{align}
[\pmb{Q}]_{k,k}=\prod_{j=0}^{k-2}b_{j\mbox{mod}N_p}.
\end{align}

\item The matrix $\pmb{P}$,
\begin{align}
\pmb{P}=\pmb{I}_{N_b}\otimes \pmb{s},
\end{align}
where, $\pmb{I}_{N_b}$ is an identical matrix, and $\pmb{s}$ is a
vector with length $N_p$ of alternating $0$, $1$,
\begin{align}
\pmb{s}=[0,1,0,1,\ldots,1]^T.
\end{align}
The vector $\pmb{r}$ is,
\begin{align}
\pmb{r}=\pmb{i}_{N_b}\otimes (\pmb{i}_{N_p}-\pmb{s}),
\end{align}
where $\pmb{i}_{N_b}$ and $\pmb{i}_{N_p}$ are the all one column
vectors with length $N_b$ and $N_p$ respectively.
\end{itemize}
}

\section{SDP Problem Formulation}
\label{section_problem_formulation}

Similarly as in \cite{ma09}, we reformulate the difficult discrete
optimization problem into a matrix optimization and relax it into an
SDP problem. The SDP formulation in this paper is slightly different
from the one in \cite{ma09}. Instead of introducing auxiliary
variables, we formulate the SDP problem with the following convex
objective function f(\pmb{U}).
\begin{align}
f(\pmb{U})=\sum_{m=1}^{N_r}
\left\{z[m]-\pmb{r}^t\pmb{Q}^t\pmb{B}[m]\pmb{Q}\pmb{r} -
\pmb{r}^t\pmb{Q}^t\pmb{B}[m]\pmb{Q}\pmb{P}\pmb{d}  \right. \nonumber \\
-\left.\pmb{r}^t\pmb{Q}^T\pmb{B}[m]^t\pmb{Q}\pmb{P}\pmb{d}
-tr\left\{\pmb{D}\pmb{P}^t\pmb{Q}^t\pmb{B}[m]\pmb{Q}\pmb{P}\right\}\right\}^2,
\end{align}
where $\pmb{U}$ is a $N_b+1$ by $N_b+1$ positive semi-definite
symmetric matrix, $\pmb{D}$ denote the sub-matrix of $\pmb{U}$
formed by selecting the last $N_b$ rows and columns, and $\pmb{d}$
is a vector $\pmb{d}=[\pmb{U}_{1,2},\ldots,\pmb{U}_{1,N_b+1}]^t$

The convex SDP problem is summarized as follows.
\begin{align}
 & \min f(\pmb{U}) \nonumber \\
  \mbox{subject to: } &  \pmb{U}_{n,n}=1, \mbox{ for
all }n,
\\
&  \pmb{U}\in \pmb{\mathcal S}, \\
&   \pmb{U}\succeq 0.
\end{align}
The demodulation result is obtained from the solution of the above
SDP problem by thresholding. That is, the demodulation result for
the $n$th symbol is obtained as $\mbox{sign}(\pmb{U}_{1,n+1})$.

\section{Approximate Semi-definite Programming Algorithm}
\label{section_approximate_sdp}

In this section, we propose a new approximate algorithm of solving
semi-definite programming. The algorithm is a generalization of
Hazan's algorithm  on approximate semi-definite programming
\cite{hazan}. Hazan's algorithm considers a special class of SDP
optimization problems, where the constraints are total trace
constraints. Such SDP optimization problems usually arise in Quantum
State Tomograph (QST) problems. The algorithm proposed in this paper
considers the class of problems with the constraints that the
diagonal elements of the matrix must be one.

We consider the following SDP optimization problem.
\begin{align}
\label{main_optimization_problem}
 & \min f(\pmb{X}) \nonumber \\
\mbox{subject to: } & \mbox{ diagonal elements of }\pmb{X} \mbox{ are zeros}, \nonumber \\
&  \pmb{X}\mbox{ is symmetric}, \nonumber \\
&   \pmb{X}+\pmb{I}\succeq 0,
\end{align}
where, $\pmb{X}$ is a square matrix, $\pmb{I}$ is the identity
matrix with the same numbers of rows and columns. Without loss of
generality, we assume that $f(\cdot)$ is independent of the diagonal
elements of the matrix $\pmb{X}$. We also assume that $f(\cdot)$ has
a bounded curvature constant $C_f$. The curvature constant $C_f$ is
defined as follows.
\begin{align}
C_f = \sup
\frac{1}{\beta^2}\left[f(\pmb{Y})-f(\pmb{X})+(\pmb{Y}-\pmb{X})^t\bigtriangledown
f(\pmb{X})\right]
\end{align}
where, $\pmb{X}+\pmb{I}\succeq 0, \pmb{Z}+\pmb{I}\succeq 0,
\pmb{Y}=\pmb{X}+\beta(\pmb{Z}-\pmb{X})$, and all diagonal elements
of $\pmb{X}$, $\pmb{Y}$, $\pmb{Z}$ are zeros. Clearly, the
convex-SDP optimization problem in the previous section can be
reduced into the above optimization problem and solved.

Before going into details of the proposed algorithm, we need some
basic facts on matrices. These facts will be presented in Section
\ref{sub_basic_fact}. The dual function of $f(\pmb{X})$ will be
discussed in Section \ref{sub_weak_duality}. The algorithm will be
presented in Section \ref{sub_algorithm}. The correctness and
convergence of the algorithm will be proved in Section
\ref{sub_correctness}. Certain discussions will be presented in
Section \ref{sub_discussion}.

\subsection{Some Basic Facts}
\label{sub_basic_fact}

\begin{lemma}
\label{basic_lemma_one} Let $\pmb{X}$ be a symmetric matrix with all
the diagonal elements being zero. Then $\pmb{I}+\pmb{X}$ is positive
semi-definite, if and only if $\lambda_{\min}(\pmb{X})\geq -1$,
where $\lambda_{\min}(\pmb{X})$ denotes the smallest eigenvalue of
$\pmb{X}$.
\end{lemma}
\begin{proof}
Necessary condition: assume that $\pmb{X}$ is positive
semi-definite, then
\begin{align}
\lambda_{\min}(\pmb{X}) & =
\min_{||\pmb{v}||=1}\pmb{v}^t\pmb{X}\pmb{v}, \nonumber \\
& = \min_{||\pmb{v}||=1}\pmb{v}^t(\pmb{I}+\pmb{X})\pmb{v} -
\pmb{v}^t\pmb{I}\pmb{v}, \nonumber \\
& \geq  \min_{||\pmb{v}||=1}0 - \pmb{v}^t\pmb{I}\pmb{v} = -1.
\end{align}

Sufficient condition: It is sufficient to show that
$\pmb{v}^t(\pmb{I}+\pmb{X})\pmb{v}\geq 0$ for all $\pmb{v}$ with
$||\pmb{v}||=1$. The above statement follows from the fact that
$\pmb{v}^t\pmb{X}\pmb{v}\geq
\lambda_{\min}(\pmb{X})||\pmb{v}||^2\geq -1$.
\end{proof}

\begin{lemma}
\label{basic_lemma_two} Let $\pmb{X}_1$ $\pmb{X}_2$ be two symmetric
matrices, such that the smallest eigenvalues of the matrices are
greater than $-1$,
\begin{align}
\lambda_{\min}(\pmb{X}_1)\geq -1,\,\,\,\lambda_{\min}(\pmb{X}_2)\geq
-1.
\end{align}
Let $\pmb{X}$ be a linear combination of $\pmb{X}_1$ and
$\pmb{X}_2$. That is $\pmb{X}= \beta \pmb{X}_1+(1-\beta)\pmb{X}_2$,
where $0\leq \beta \leq 1$. Then, the smallest eigenvalue of
$\pmb{X}$ is also greater than $-1$,
\begin{align}
\lambda_{\min}(\pmb{X})\geq -1.
\end{align}
\end{lemma}
\begin{proof}
\begin{align}
\lambda_{\min}(\pmb{X}) & = \min_{||\pmb{v}||=1}
\pmb{v}^t\pmb{X}\pmb{v} \nonumber \\
& = \min_{||\pmb{v}||=1}
\pmb{v}^t(\beta\pmb{X}_1+(1-\beta)\pmb{X}_2)\pmb{v}\nonumber \\
& \geq  \beta(-1)+(1-\beta)(-1) = -1.
\end{align}
\end{proof}

\subsection{Weak Duality}

\label{sub_weak_duality}

The proposed algorithm is based on iteratively reducing the duality
gap between the primal function and its dual function. For a primal
function $f(\pmb{X})$, we define the dual function $w(\pmb{X})$ as
\begin{align}
w(\pmb{X}) & =  \max_{ \bigtriangledown
f(\pmb{X})+\pmb{\lambda}\succeq 0 } w(\pmb{X},\pmb{\lambda})
\nonumber \\
& = \max_{ \bigtriangledown f(\pmb{X})+\pmb{\lambda}\succeq 0 }
f(\pmb{X}) - \pmb{X}\cdot \bigtriangledown f(\pmb{X}) -
tr(\pmb{\lambda}),
\end{align}
where, $\pmb{\lambda}=
\mbox{diag}(\lambda_1,\lambda_2,\ldots,\lambda_n)$ is a diagonal
matrix.

\begin{theorem}(\emph{Weak Duality}) Denote the minimizer of
the optimization problem in Eq. \ref{main_optimization_problem} as
$\pmb{X}^\ast$. Let $\pmb{X}$ be a feasible point.
 Then, the following weak duality inequalities hold.
\begin{align}
f(\pmb{X})\geq f(\pmb{X}^\ast)\geq w(\pmb{X})
\end{align}
\end{theorem}
\begin{proof}
Given a function $f(\pmb{X})$, the corresponding Lagrangian function
can be written as
\begin{align}
f(\pmb{X})-\pmb{V}\cdot (\pmb{I}+\pmb{X})+\pmb{\lambda}\cdot \pmb{X}
\end{align}
where $\pmb{V}$ is a symmetric positive semi-definite matrix.

We can rewrite the function $f(\pmb{X}^\ast)$ in a min-max form as
follows.
\begin{align}
& f(\pmb{X}^\ast) = \min_{\pmb{X}\succeq 0, \pmb{X}\in \mathcal{S}}
f(\pmb{X})
\nonumber \\
&= \min_{\pmb{X}\in \mathcal{S}}\left[ \max_{\pmb{V}\succeq 0,
\pmb{\lambda}}\left[ f(\pmb{X})-\pmb{V}\cdot
(\pmb{I}+\pmb{X})+\pmb{\lambda}\cdot \pmb{X}\right]\right]
\end{align}
This is because
\begin{align}
&  \max_{\pmb{V}\succeq 0, \pmb{\lambda}}\left[
f(\pmb{X})-\pmb{V}\cdot (\pmb{I}+\pmb{X})+\pmb{\lambda}\cdot
\pmb{X}\right]
\nonumber \\
& = \left\{
\begin{array}{ll}
f(\pmb{X}), & \mbox{ if }\pmb{X}\succeq 0 \mbox{ and}\\
            & \mbox{ diagonal elements of }\pmb{X}\mbox{ are zeros}     \\
+\infty,    & \mbox{ otherwise }
\end{array}
\right.
\end{align}

By the max-min inequality (see for example, \cite{boyd} page 238,
Eq. 5.47), we can lower bound $f(\pmb{X}^\ast)$ as follows.
\begin{align}
& f(\pmb{X}^\ast) =  \min_{\pmb{X}\in \mathcal{S}}
\left[\max_{\pmb{V}\succeq 0, \pmb{\lambda}}\left[
f(\pmb{X})-\pmb{V}\cdot (\pmb{I}+\pmb{X})+\pmb{\lambda}\cdot
\pmb{X}\right]\right]
\nonumber \\
& \geq \max_{\pmb{V}\succeq 0, \pmb{\lambda}} \left[\min_{\pmb{X}\in
\mathcal{S}} \left[ f(\pmb{X})-\pmb{V}\cdot
(\pmb{I}+\pmb{X})+\pmb{\lambda}\cdot \pmb{X}\right]\right]
\end{align}

Let us assume that $\pmb{V}_0$ and $\pmb{\lambda}_0$ are symmetric
and diagonal matrix respectively, such that the following equations
hold for a feasible point $\pmb{X}_0$.
\begin{align}
\bigtriangledown f(\pmb{X}_0)-\pmb{V}_0+\pmb{\lambda}_0=\pmb{0},
\\
\bigtriangledown f(\pmb{X}_0)+\pmb{\lambda}_0 \succeq 0.
\end{align}

By the above discussions, we have
\begin{align}
&  f(\pmb{X}^\ast) \nonumber \\
& \geq \max_{\pmb{V}\succeq 0, \pmb{\lambda}}\left[ \min_{\pmb{X}\in
\mathcal{S}} \left[ f(\pmb{X})-\pmb{V}\cdot
(\pmb{I}+\pmb{X})+\pmb{\lambda}\cdot \pmb{X}\right]\right],
\nonumber \\
& \geq  \min_{\pmb{X}\in \mathcal{S}} \left[
f(\pmb{X})-\pmb{V}_0\cdot (\pmb{I}+\pmb{X})+\pmb{\lambda}_0\cdot
\pmb{X}\right],
\nonumber \\
& \stackrel{(a)}{=}   f(\pmb{X}_0)-\pmb{V}_0\cdot
(\pmb{I}+\pmb{X}_0)+\pmb{\lambda}_0\cdot \pmb{X}_0, \nonumber \\
& \stackrel{(b)}{=}   f(\pmb{X}_0)-(\bigtriangledown
f(\pmb{X}_0)+\pmb{\lambda}_0)\cdot
(\pmb{I}+\pmb{X}_0)+\pmb{\lambda}_0\cdot \pmb{X}_0, \nonumber \\
& =   f(\pmb{X}_0)-\bigtriangledown f(\pmb{X}_0)\cdot \pmb{X} -
\pmb{\lambda}_0\cdot \pmb{I},
\end{align}
where, (a) follows from the fact that $\pmb{X}_0$ is exactly the
minimizer, and (b) follows from the definition of $\pmb{V}_0$.
Therefore,
\begin{align}
f(\pmb{X}^\ast) & \geq \max_{ \bigtriangledown
f(\pmb{X}_0)+\pmb{\lambda}\succeq 0 } f(\pmb{X}_0) - \pmb{X}_0\cdot
\bigtriangledown f(\pmb{X}_0) - tr(\pmb{\lambda}), \nonumber \\
& \geq w(\pmb{X}_0).
\end{align}
The theorem follows from the fact that $\pmb{X}_0$, $\pmb{V}_0$, and
$\pmb{\lambda}_0$ are arbitrary.
\end{proof}

The above weak duality theorem provides a way to estimate how far a
feasible point $\pmb{X}$ is away from the optimal solution. We
define
\begin{align}
h(\pmb{X}) = f(\pmb{X})-f(\pmb{X}^\ast), \\
g(\pmb{X}) = f(\pmb{X})-w(\pmb{X}).
\end{align}
By the weak duality theorem, we have $h(\pmb{X})\leq g(\pmb{X})$.

In order to evaluate the dual function $w(\pmb{X})$, the following
optimization problem needs to be solved.
\begin{align}
\label{dual_maximization_one}
& \min \sum_i \lambda_i \nonumber \\
& \mbox{subject to } \gamma_i\geq 0, \mbox{ for all }i
\end{align}
where $\gamma_i$ is the $i$th eigenvalue of the matrix
$\bigtriangledown f(\pmb{X})+\pmb{\lambda}$.

\begin{lemma}
\label{gradient_calculation_lemma} Let $\gamma_i$ denote the $i$th
eigenvalue of the matrix $\bigtriangledown
f(\pmb{X})+\pmb{\lambda}$. Let $\pmb{v}_i$ denote the corresponding
eigenvectors. Then,
\begin{align}
\Delta \gamma_i = \sum_j (v_{ij})^2 \Delta \lambda_i
\end{align}
where, $\Delta \gamma_i$ and $\Delta \lambda_i$ are the
infinitesimal differences, $v_{ij}$ denotes the $j$th element of the
vector $\pmb{v}_i$.
\end{lemma}
\begin{proof}
It is clear that there exists a decomposition of $\bigtriangledown
f(\pmb{X})+\pmb{\lambda}$,
\begin{align}
\bigtriangledown f(\pmb{X})+\pmb{\lambda} =
\pmb{V}\pmb{\Lambda}\pmb{V}^t
\end{align}
such that $\pmb{V}$ is a unitary matrix and $\Lambda$ is a diagonal
matrix. In fact, $\pmb{V}=[\pmb{v}_1,\pmb{v}_2,\ldots,\pmb{v}_n]$
and $\Lambda = \mbox{diag}(\gamma_i)$ is a such decomposition.

Let $\Delta \pmb{V}$ and $\Delta \pmb{\Lambda}$ be the corresponding
infinitesimal differences of  $\pmb{V}$ and $\pmb{\Lambda}$
respectively. Then, we have
\begin{align}
\label{diff_trick_eq1} (\pmb{V}+\Delta \pmb{V})
(\pmb{\Lambda}+\Delta \pmb{\Lambda}) (\pmb{V}^t+\Delta \pmb{V}^t) =
\triangledown f(\pmb{X}) + \pmb{\lambda} + \Delta \pmb{\lambda},
\end{align}
\begin{align}
\label{diff_trick_eq2} (\pmb{V}^t+\Delta \pmb{V}^t)(\pmb{V}+\Delta
\pmb{V})
 = \pmb{I}.
\end{align}
From Eq. \ref{diff_trick_eq2} and the fact that $\pmb{V}$ is
unitary, we have
\begin{align}
\pmb{V}^t \Delta \pmb{V} +\Delta \pmb{V}^t \pmb{V}=0.
\end{align}
Since $\pmb{V}^t\Delta \pmb{V}$ and $\Delta \pmb{V}^t \pmb{V}$ are
the transpose of each other, we conclude that the matrices
$\pmb{V}^t\Delta \pmb{V}$ and $\Delta \pmb{V}^t \pmb{V}$ are
anti-symmetric and their diagonal elements are all zeros.

From Eq. \ref{diff_trick_eq1}, we have
\begin{align}
\pmb{V}\Delta \pmb{\Lambda} \pmb{V}^t + \Delta\pmb{V} \pmb{\Lambda}
\pmb{V}^t + \pmb{V} \pmb{\Lambda} \Delta\pmb{V}^t = \Delta
\pmb{\lambda}.
\end{align}
Multiplying the above equation by the matrix $\pmb{V}^t$ at the left
side and the matrix $\pmb{V}$ at the right side, we obtain
\begin{align}
\Delta \pmb{\Lambda}  + \pmb{V}^t\Delta\pmb{V} \pmb{\Lambda} +
\pmb{\Lambda} \Delta\pmb{V}^t \pmb{V}= \pmb{V}^t \Delta
\pmb{\lambda} \pmb{V}.
\end{align}
Since the diagonal elements of the matrices $\pmb{V}^t\Delta
\pmb{V}$ and $\Delta \pmb{V}^t \pmb{V}$ are all zeros, the diagonal
elements of the matrices $\pmb{V}^t\Delta\pmb{V} \pmb{\Lambda}$ and
$\pmb{\Lambda} \Delta\pmb{V}^t \pmb{V}$ are also zeros. Therefore,
we conclude that the diagonal elements of $\Delta \pmb{\Lambda}$ and
$\pmb{V}^t \Delta \pmb{\lambda} \pmb{V}$ are identical. The theorem
then follows from the fact that the $i$th diagonal element of the
matrix $\pmb{V}^t \Delta \pmb{\lambda} \pmb{V}$ is $\sum_j
({v}_{ij})^2 \Delta \lambda_j$.
\end{proof}

\begin{lemma}
\label{lemma_def_y} In the optimization problem in Eq.
\ref{dual_maximization_one}. Let $\pmb{\lambda}^\ast$ be the
minimizer. Let $\gamma_i$ denote the $i$th eigenvalue of the matrix
$\bigtriangledown f(\pmb{X})+\pmb{\lambda}^\ast$. Let $\pmb{v}_i$
denote the corresponding eigenvectors. Then, there exist a set
${\mathcal T}\subset \{1,2,\ldots,n\}$ and a vector $\pmb{y}$, such
that
\begin{align}
\label{simple_sdp_char_one} \pmb{y}\geq 0,
\end{align}
\begin{align}
\label{simple_sdp_char_two} \pmb{V}^t \pmb{y}=[1,\ldots,1]^t,
\end{align}
\begin{align}
\label{simple_sdp_char_three} \pmb{v}_i^t\left(\bigtriangledown
f(\pmb{X}) + \pmb{\lambda}^\ast\right)\pmb{v}_i=0, \mbox{ for all
}i\in {\mathcal T},
\end{align}
where $n$ is the number of rows of matrix $\pmb{\lambda}$, $\pmb{V}$
is a matrix such that each row of $\pmb{V}$ is $[v_{ij}^2]$ for one
$i\in {\mathcal T}$. That is,
\begin{align}
\pmb{V} = \left[ \begin{array}{cccc}
v_{i_1,1}^2 & v_{i_2,2}^2 & \ldots & v_{i_1,n}^2 \\
\ldots      & \ldots      & \ldots &  \ldots     \\
v_{i_k,1}^2 & v_{i_k,2}^2 & \ldots & v_{i_k,n}^2
\end{array}
\right],
\end{align}
where $i_1,\ldots, i_k\in{\mathcal T}$.
\end{lemma}
\begin{proof}
Due to the nature of the optimization problem, there exist at least
one active constraint at the minimizer. We say that an inequality
constraint is active at a feasible point, if the inequality
constraint holds with equality. In this optimization problem, the
$i$th inequality constraint is active, if $\gamma_i=0$. Let
$\mathcal T$ denote the set of indexes of all active constraints.
Then, for all $i\in {\mathcal T}$, $\gamma_i=0$,
\begin{align}
\pmb{v}_i^t\left(\bigtriangledown f(\pmb{X}) +
\pmb{\lambda}^\ast\right)\pmb{v}_i= 0.
\end{align}

Due to the Karush-Kuhn-Tucker (KKT) Theorem ( see \cite{cong01}
Theorem 20.1 .  Page 398), there exists a vector $\pmb{y}$ such that
\begin{align}
\pmb{y}\geq 0,
\end{align}
\begin{align}
\label{temp_64} \sum_i y_i \bigtriangledown \gamma_i =
\bigtriangledown { \sum_i \lambda_i} = [1,1,\ldots,1]^t.
\end{align}
According to Lemma \ref{gradient_calculation_lemma},
$\bigtriangledown \gamma_i = [v_{i1}^2,v_{i2}^2, \ldots,
v_{in}^2]^t$. Therefore
\begin{align}
\label{temp_64} \sum_i y_i \bigtriangledown \gamma_i = \pmb{V}^t
\pmb{y}
\end{align}
The lemma follows.
\end{proof}

\begin{corollary}
For $i\in {\mathcal T}$, define $ \alpha_i =
\pmb{v}_i^t\bigtriangledown f(\pmb{X})\pmb{v}_i $. Then,
\begin{align}
\sum_i \lambda_i^\ast  = -\sum_{i \in {\mathcal T}} y_{i}\alpha_i
\end{align}
\end{corollary}
\begin{proof}
\begin{align}
\sum_i \lambda_i^\ast & = [1,\ldots,1]\mbox{diag}(\pmb{\lambda}^\ast) = \pmb{y}^t \pmb{V} (\mbox{diag}(\pmb{\lambda}^\ast)) \nonumber \\
& = \pmb{y}^t \left[\ldots, \sum_j (v_{ij})^2\lambda_j, \ldots
\right]^t \nonumber \\
& = \pmb{y}^t \left[\ldots, \pmb{v_i}^t\pmb{\lambda}^\ast\pmb{v}_i,
\ldots
\right]^t \nonumber \\
& \stackrel{(a)}{=} \pmb{y}^t \left[\ldots, -\alpha_i, \ldots
\right]^t \nonumber \\
& = -\sum_{i\in {\mathcal T}} y_i\alpha_i,
\end{align}
where, $\mbox{diag}(\pmb{\lambda}^\ast)$ denote the column vector
that consists of diagonal elements of $\pmb{\lambda}^\ast$, and (a)
follows from Eq. \ref{simple_sdp_char_three}.
\end{proof}

\subsection{The Algorithm}

\label{sub_algorithm}

The proposed algorithm is summarized as follows.

\begin{itemize}

\item Step 1: set k=1, set $\pmb{X}$ to a feasible point;

\item Step 2: calculate the gradient $\bigtriangledown
f(\pmb{X})$;

\item Step 3: solve the optimization problem in Eq.
\ref{dual_maximization_one}, obtain $\alpha_i$, $y_i$, $\pmb{v}_i$
for $i\in {\mathcal T}$;

\item Step 4: calculate the function $g(\pmb{X})$, if  $g(\pmb{X})$
is less than a certain threshold, go to step 8, otherwise, go to the
next step;

\item Step 5: update $\Delta \pmb{X}$ as follows,
\begin{align}
\Delta \pmb{X} = \beta_k \left(\left(\sum_{i\in {\mathcal T}}
y_i\pmb{v}_i\pmb{v}_i^t\right) - \pmb{X} - \pmb{I}\right)
\end{align}
where, $\beta_k$ is a predefined step size parameter;

\item Step 6: update $\pmb{X}=\pmb{X}+\Delta \pmb{X}$;

\item Step 7: set k=k+1, go to step 2;

\item Step 8: return $\pmb{X}$, stop.

\end{itemize}

\subsection{Correction and Convergence}

\label{sub_correctness}

In this subsection, we show that the proposed algorithm is correct
and converges.

\begin{theorem}
\label{main_theorem_one} In the proposed algorithm, the diagonal
elements of $\pmb{X}$ are zeros, and $\lambda_{\min}(\pmb{X})\geq
-1$.
\end{theorem}
\begin{proof}
Prove by induction. It is sufficient to show that $\pmb{X}+\Delta
\pmb{X}$ satisfies the above conditions, if $\pmb{X}$ satisfies the
conditions.

Note that the $k$th diagonal element of the matrix $\sum_{i \in
{\mathcal T}} y_i \pmb{v}_i\pmb{v}_i^t$ is equal to $\sum_{i \in
{\mathcal T}} y_i v_{ik}^2$, is also equal to the $k$th element of
$\pmb{V}^t \pmb{y}$. From Lemma \ref{lemma_def_y}, we have that the
$k$th element of $\pmb{V}^t \pmb{y}$ is one. Therefore, the diagonal
elements of the matrix $\sum_{i \in {\mathcal T}} y_i
\pmb{v}_i\pmb{v}_i^t-\pmb{X} - \pmb{I}$ are all zeros. The diagonal
elements of the matrix $\pmb{X}+\Delta\pmb{X}$ are also all zeros.

We can show that $\lambda_{\min}(\pmb{X}+\Delta\pmb{X})\geq -1$, if
we can show that
\begin{align}
\label{temp_69}
 \lambda_{\min}\left[\left(\sum_{i\in {\mathcal T}} y_i\pmb{v}_i\pmb{v}_i^t\right) -
 \pmb{I}\right]\geq -1,
 \end{align}
\begin{align}
\label{temp70} \lambda_{\min}\left(\pmb{X}\right)\geq -1.
\end{align}
This is because of Lemma \ref{basic_lemma_two} and
$\pmb{X}+\Delta\pmb{X}$ being a linear combination of the above two
matrices
\begin{align}
\pmb{X}+\Delta\pmb{X} = \beta_k \left(\sum_{i\in {\mathcal T}}
y_i\pmb{v}_i\pmb{v}_i^t - \pmb{I}\right) + (1-\beta_k)\pmb{X}
\end{align}
Eq. \ref{temp70} follows from the given hypothesis. Eq.
\ref{temp_69} follows from Lemma \ref{basic_lemma_one}, and
$\sum_{i\in {\mathcal T}} y_i\pmb{v}_i\pmb{v}_i^t$ being positive
semi-definite.
 The theorem is proven.
\end{proof}

\begin{theorem}
\label{main_theorem_two} In the proposed optimization algorithm, let
$\pmb{X}_k$ denote the value of $\pmb{X}$ after $k$ iterations.
Then, the gap
\begin{align}
 h(\pmb{X}_{k+1})  & \leq (1-\beta_k)h(\pmb{X}_k) +
\beta_k^2 C_f.
\end{align}
Therefore,  $h(\pmb{X}_k)$ goes to zero, and $f(\pmb{X}_k)$ goes to
$f(\pmb{X}^\ast)$, for properly chosen step size parameters
$\beta_k$.
\end{theorem}
\begin{proof}
First, we wish to show that the following equality holds.
\begin{align}
\label{useful_reasoning}
 \left(\sum_{i\in {\mathcal T}}y_i\pmb{v}_i\pmb{v}_i^t\right)\cdot
\bigtriangledown f(\pmb{X})  = \sum_{i\in {\mathcal T}}y_i\alpha_i
\end{align}
The reasoning is as follows.
\begin{align}
& \left(\sum_{i\in {\mathcal T}}y_i\pmb{v}_i\pmb{v}_i^t\right)\cdot
\bigtriangledown f(\pmb{X})  = \sum_{i\in {\mathcal
T}}\left(y_i\pmb{v}_i\pmb{v}_i^t\cdot \bigtriangledown
f(\pmb{X})\right) \nonumber \\
& = \sum_{i\in {\mathcal T}}\left(tr\left(y_i\pmb{v}_i\pmb{v}_i^t
\bigtriangledown
f(\pmb{X})\right)\right) \nonumber \\
& \stackrel{(a)}{=} \sum_{i\in {\mathcal
T}}y_i\left(tr\left(\pmb{v}_i^t \bigtriangledown
f(\pmb{X})\pmb{v}_i\right)\right) \nonumber \\
& = \sum_{i\in {\mathcal T}}y_i\left(\pmb{v}_i^t \bigtriangledown
f(\pmb{X})\pmb{v}_i\right) \nonumber \\
& \stackrel{(b)}{=} \sum_{i\in {\mathcal T}}y_i\alpha_i
\end{align}
where, (a) follows from the the property of trace,
$tr(\pmb{A}\pmb{B})=tr(\pmb{B}\pmb{A})$ for all matrices $\pmb{A}$
and $\pmb{B}$, and (b) follows from the definition of $\alpha_i$.

The value of $f(\pmb{X}_{k+1})$ can be upper bounded as follows.
\begin{align}
& f(\pmb{X}_{k+1}) = f\left(\pmb{X}_k+\beta_k\left(\sum_{i\in
{\mathcal T}} y_i \pmb{v}_i\pmb{v}_i^t -
 \pmb{X}_k-\pmb{I}\right)\right) \nonumber \\
& \stackrel{(a)}{\leq} f(\pmb{X}_k) + \beta_k\left(\sum_{i\in
{\mathcal T}} y_i\pmb{v}_i\pmb{v}_i^t - \pmb{X}_k
-\pmb{I}\right)\cdot\bigtriangledown f(\pmb{X}_k) + \beta_k^2 C_f
\nonumber \\
& \stackrel{(b)}{=} f(\pmb{X}_k) + \beta_k\left(\left(\sum_{i \in
{\mathcal T}}y_i\pmb{v}_i\pmb{v}_i^t\right)\cdot \bigtriangledown
f(\pmb{X}_k) - \pmb{X}_k \cdot \bigtriangledown f(\pmb{X}_k)\right)
\nonumber \\
& \hspace{0.5in} + \beta_k^2 C_f \nonumber \\
& \stackrel{(c)}{=} f(\pmb{X}_k) +\beta_k\left(\sum_{i\in {\mathcal
T}}y_i\alpha_i - \pmb{X}_k \cdot
\bigtriangledown f(\pmb{X}_k)\right)+ \beta_k^2 C_f \nonumber \\
& \stackrel{(d)}{=} f(\pmb{X}_k)
+\beta_k\left(-\sum_{i}\lambda_i^\ast - \pmb{X}_k \cdot
\bigtriangledown f(\pmb{X}_k)\right)+ \beta_k^2 C_f \nonumber \\
& = f(\pmb{X}_k) -\beta_k g(\pmb{X}_k) + \beta_k^2 C_f,
\end{align}
where, (a) follows from the definition of $C_f$, (b) follows from
the fact that the diagonal elements of $\bigtriangledown f(\pmb{X})$
are all zeros, $\pmb{I}\cdot\bigtriangledown f(\pmb{X})=0$, (c)
follows from Eq. \ref{useful_reasoning}, and (d) follows from Eq.
\ref{simple_sdp_char_three}.

 Therefore, we have
\begin{align}
 h(\pmb{X}_{k+1})  & = f(\pmb{X}_{k+1}) - f(\pmb{X}^\ast) \nonumber
 \\
& \leq f(\pmb{X}_{k}) - f(\pmb{X}^\ast) - \beta_k g(\pmb{X}_k) +
\beta_k^2C_f \nonumber \\
& \leq h(\pmb{X}_{k}) - \beta_k g(\pmb{X}_k) + \beta_k^2C_f
\nonumber \\
& \leq h(\pmb{X}_{k}) - \beta_k h(\pmb{X}_k) +
\beta_k^2C_f \nonumber \\
& \leq (1-\beta_k)h(\pmb{X}_k) + \beta_k^2 C_f.
\end{align}
The theorem follows.
\end{proof}

\subsection{Discussion}

\label{sub_discussion}

One character of the proposed algorithm is that a close approximate
solution can be found after only a few iterations. By Theorem
\ref{main_theorem_two}, we can see that the optimal step size
parameter $\beta_k$ at the $k$th iteration depends on the current
gap $h(\pmb{X}_k)$ and $C_f$. At the first several iterations, the
parameter $\beta_k$ can take larger values, and the gap
$h(\pmb{X}_k)$ decreases quickly.

Because the solution of the SDP optimization problem is an
intermediate result in the demodulation and equalization algorithm,
approximate solutions are usually sufficient to ensure that the
demodulation results are correct with high probability. In fact, we
find that only few iterations are usually needed to ensure low
demodulation error probability by simulation results.

During each iteration, one optimization problem needs to be solved
to calculate the dual function. However, compared with the original
matrix optimization problem with approximately $n^2$ optimization
variables, the optimization problem in dual function calculation
only has $n$ optimization variables. Therefore, the optimization
problem in each iteration can be solved with lower computational
complexity and storage requirements.

Overall, the proposed algorithm has lower computational complexity
and storage requirements. It is an attractive choice for high-speed
real-time demodulation implementations.

\section{Numerical Results }

\label{sec_numerical}

In this section, we present simulation results for the proposed
demodulation and equalization scheme with approximate SDP
programming. We assume that the transmitted pulses are the second
derivative Gaussian monocycles,
\begin{align}
\bar{w}(t)=\left[1-4\pi\left(t/\tau_m\right)^2\right]
\exp\left\{-2\pi \left(t/\tau_m\right)^2\right\},
\end{align}
where $\tau_m=0.2877$ nanosecond. Each information bearing signal
block consists of $N_b=10$ symbols and each symbol corresponds to
$N_p=4$ pulses. The symbol duration $T_s=8$ nanoseconds.

We use the IEEE 802.15.4a channel models as described in
\cite{molisch04}. Two types of channel models CM1 and CM6 are used
for simulation. We illustrate the bit error probability of the
proposed demodulation and equalization algorithm in the case of CM6
channel models in Fig. \ref{fig_long_bit_error}. A typical channel
impulse response of the CM6 model is shown in Fig.
\ref{fig_long_impulse}. We illustrate the bit error probability of
the proposed demodulation and equalization algorithm in the case of
CM1 channel models in Fig. \ref{fig_short_bit_error}. A typical
channel impulse response of the CM1 model is shown in Fig.
\ref{fig_short_impulse}. The numerical results show that the
proposed demodulation algorithm has satisfactory bit error
probability performance.

\begin{figure}[h]
 \centering
 \includegraphics[width=3.5in]{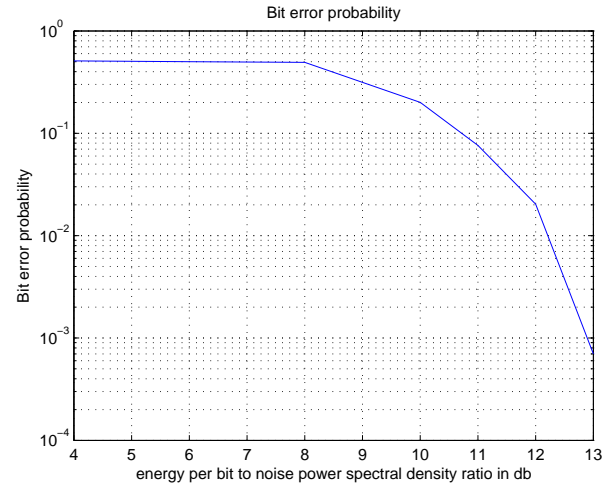}
 \caption{Bit error probabilities for the CM6 channel model. The X-axis shows
 energy per bit to noise power spectral density ratio $E_b/N_0$ in
 dB}
 \label{fig_long_bit_error}
\end{figure}

\begin{figure}[h]
 \centering
 \includegraphics[width=3.5in]{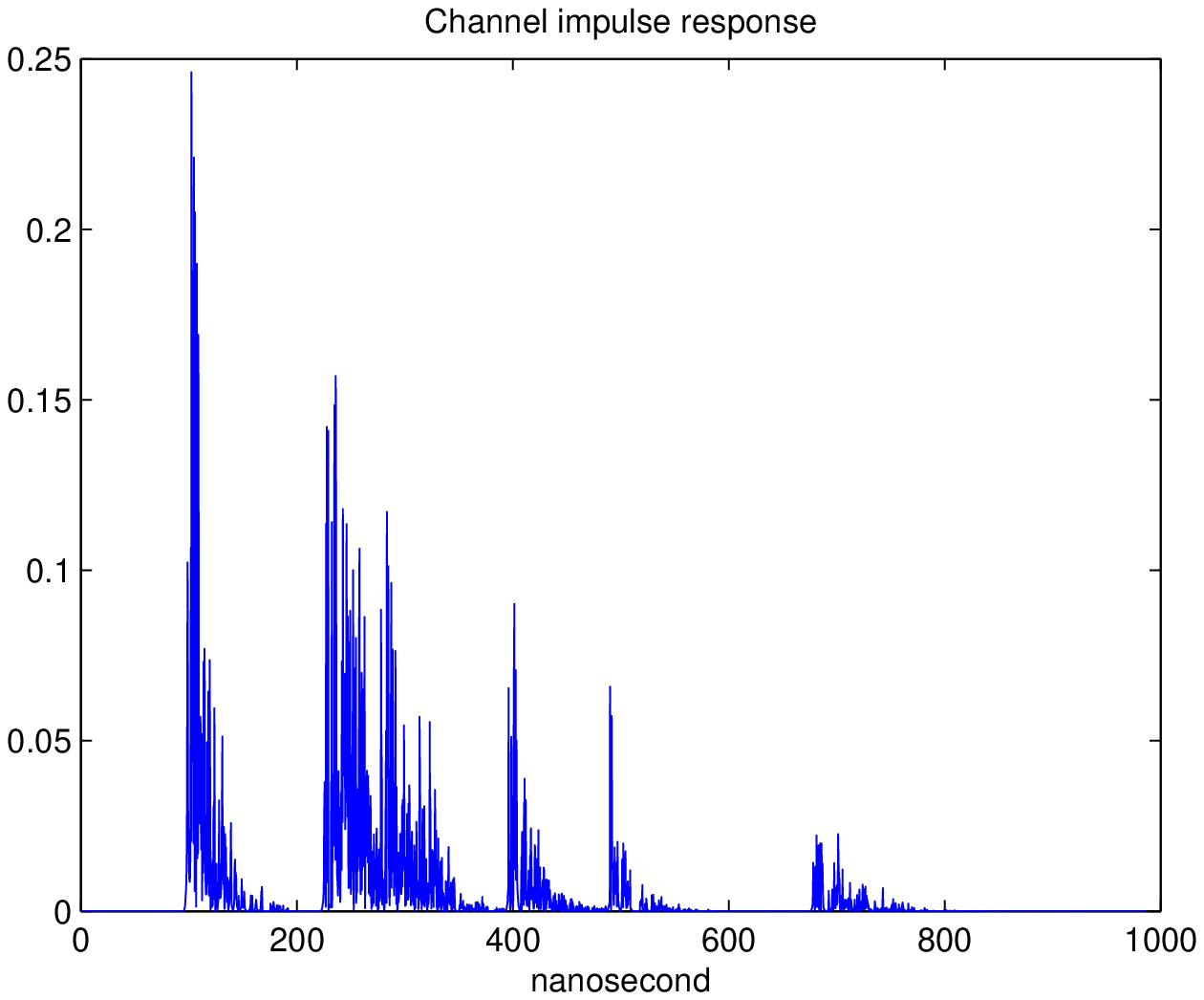}
 \caption{Typical channel impulse response in the CM6 channel model}
 \label{fig_long_impulse}
\end{figure}

\begin{figure}[h]
 \centering
 \includegraphics[width=3.5in]{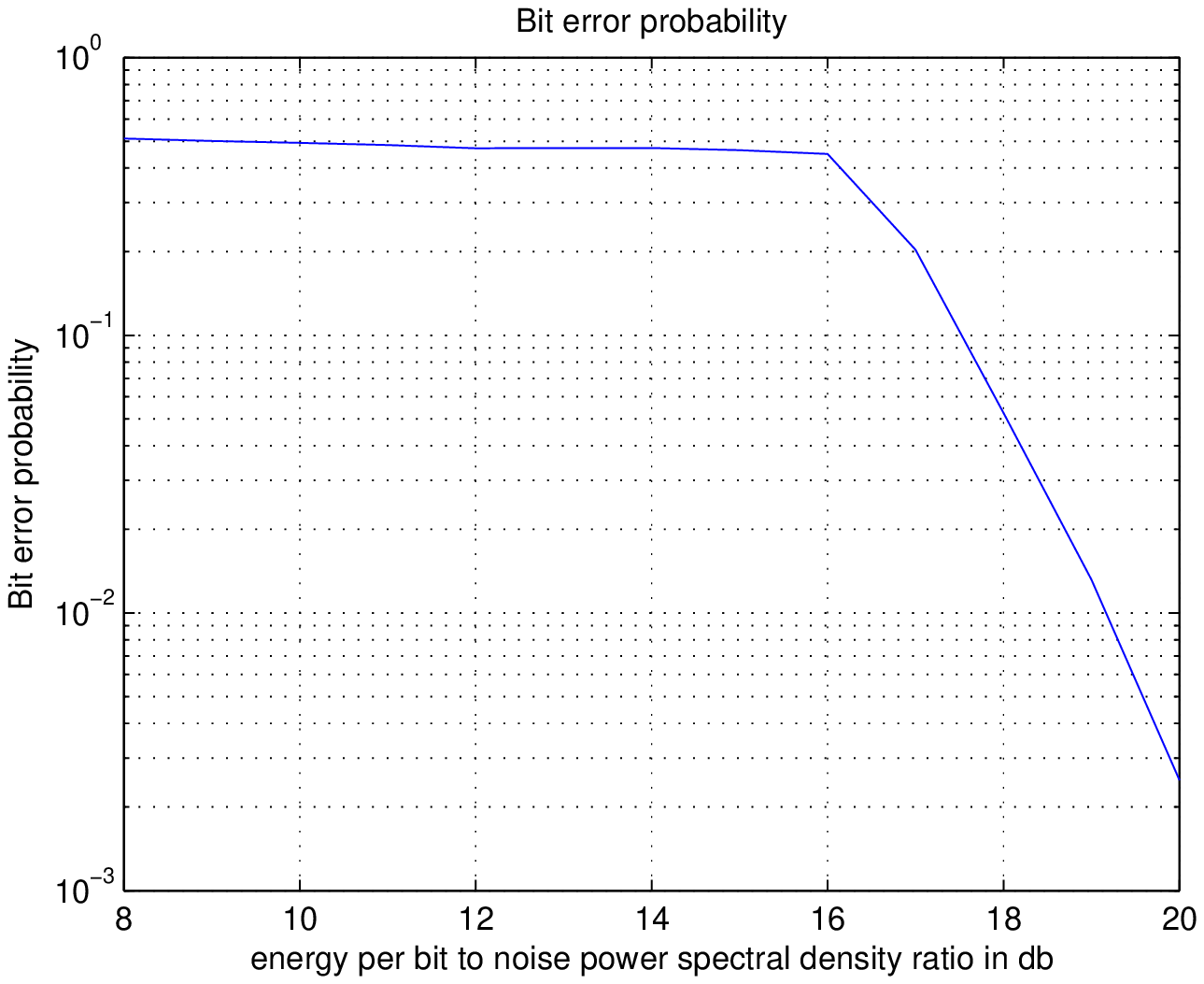}
 \caption{Bit error probabilities for the CM1 channel model. The X-axis shows
 energy per bit to noise power spectral density ratio $E_b/N_0$ in
 dB
 }
 \label{fig_short_bit_error}
\end{figure}

\begin{figure}[h]
 \centering
 \includegraphics[width=3.5in]{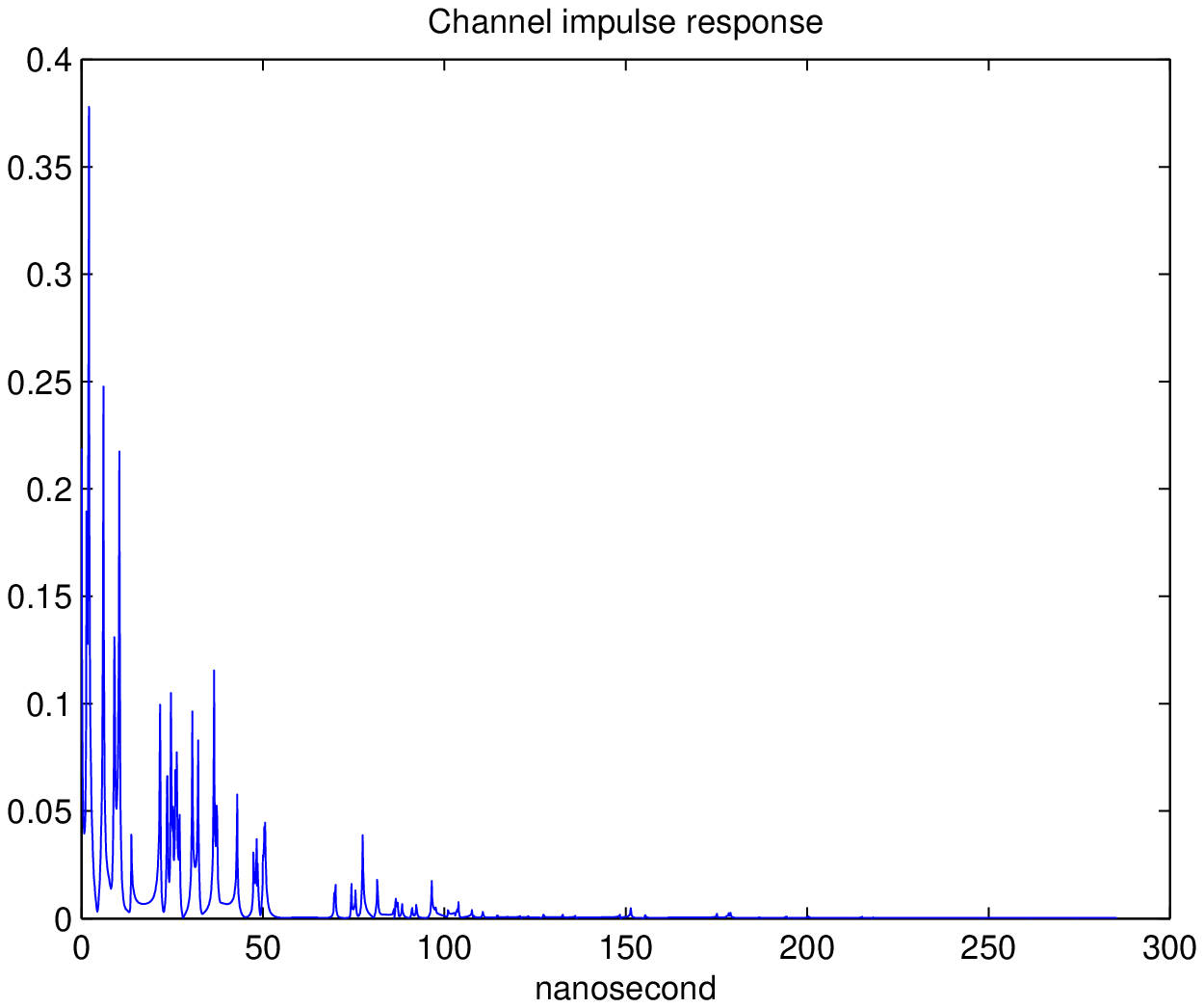}
 \caption{Typical channel impulse response in the CM6 channel model
 }
 \label{fig_short_impulse}
\end{figure}

\section{Conclusion}

\label{sec_conclusion}

In this paper, we propose an approximate semi-definite programming
framework for demodulation and equalization of non-coherent UWB
systems with inter-symbol-interference. The proposed algorithm has
low computational complexity and storage requirements, which make it
an attractive choice for real-time high-speed implementations.
Numerical results show that the proposed approach has satisfactory
error probability performance. The proposed approach can be adopted
in a wide spectrum of non-coherent UWB modulation schemes.

\bibliographystyle{IEEEtran}
\bibliography{approximate_decoding}

\end{document}